\documentclass{article}

\usepackage{arxiv}
\usepackage[english]{babel}
\usepackage[utf8]{inputenc} 
\usepackage[T1]{fontenc}    
\usepackage{hyperref}       
\usepackage{url}            
\usepackage{booktabs}       
\usepackage{amsfonts}       
\usepackage{nicefrac}       
\usepackage{microtype}      
\usepackage{lipsum}         
\usepackage{graphicx}
\usepackage[all,cmtip]{xy}
\usepackage{natbib}
\usepackage{doi}
\usepackage{float}
\usepackage{MnSymbol}
\usepackage{enumitem}
\usepackage{array}
\usepackage{amsthm}
\usepackage{caption} 
\usepackage{cleveref}       
\title{Reliability on QR Codes and Reed-Solomon Codes}

\newcolumntype{L}{>{$}l<{$}}
\newcolumntype{R}{>{$}r<{$}}
\newcolumntype{C}{>{$}c<{$}}

\date{}

\newif\ifuniqueAffiliation
\uniqueAffiliationtrue

\ifuniqueAffiliation 
\author{ \href{https://orcid.org/0009-0008-6032-8223}{\includegraphics[scale=0.06]{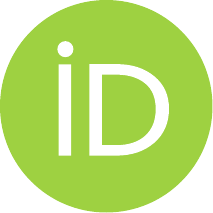}\hspace{1mm}Bhavuk Sikka Bajaj}\thanks{Learn more at \url{https://github.com/Bubbasm/}} \\
	Canary Islands, Spain \\
	\texttt{bsikka100@gmail.com} \\
}
\else
\usepackage{authblk}

\newbox{\orcid}\sbox{\orcid}{\includegraphics[scale=0.06]{orcid.pdf}} 
\author[1]{%
	\href{https://orcid.org/0000-0000-0000-0000}{\usebox{\orcid}\hspace{1mm}David S.~Hippocampus\thanks{\texttt{hippo@cs.cranberry-lemon.edu}}}%
}
\author[1,2]{%
	\href{https://orcid.org/0000-0000-0000-0000}{\usebox{\orcid}\hspace{1mm}Elias D.~Striatum\thanks{\texttt{stariate@ee.mount-sheikh.edu}}}%
}
\affil[1]{Department of Computer Science, Cranberry-Lemon University, Pittsburgh, PA 15213}
\affil[2]{Department of Electrical Engineering, Mount-Sheikh University, Santa Narimana, Levand}
\fi

\renewcommand{\shorttitle}{Reliability on QR Codes and Reed-Solomon Codes}

\newtheorem{theorem}{Theorem}[section]

\newtheorem{proposition}[theorem]{Proposition}

\theoremstyle{definition}
\newtheorem*{remark}{Remark}
\newtheorem{definition}[theorem]{Definition}
\newtheorem{example}[theorem]{Example}

\hypersetup{
pdftitle={Reliability on QR Codes and Reed-Solomon Codes},
pdfsubject={cs.CR},
pdfauthor={Bhavuk Sikka Bajaj},
pdfkeywords={QR codes \and Reed-Solomon codes \and Information manipulation},
}

\begin{document}
\maketitle

\begin{abstract}
This study addresses the use of Reed-Solomon error correction codes in QR codes to enhance resilience against failures. To fully grasp this approach, a basic cryptographic context is provided, necessary for understanding Reed-Solomon codes. The study begins by defining a code and explores key outcomes for codes with additional properties, such as linearity. The theoretical framework is further developed with specific definitions and examples of Reed-Solomon codes, presented as a particular variant of BCH codes.

Additionally, the structure of QR codes is analyzed, encompassing different versions and how data is represented in the form of black and white pixels within an image.

Finally, an inherent vulnerability of Reed-Solomon Codes, and particularly of QR codes, related to selective manipulation of modules is examined. This vulnerability leverages the error correction mechanisms present in Reed-Solomon codes.
\end{abstract}


\section{Introduction}
The present work focuses on a critical aspect of QR code (Quick Response) technology, a widely used tool today for storing and transmitting information. QR codes, consisting of a two-dimensional array of modules (square pixels), have become ubiquitous in our daily lives. They are used in various applications ranging from advertising to inventory management, product authentication, and access to digital content.

The fundamental motivation behind this project lies in the need to thoroughly understand the potential security risks associated with their use. In particular, we will focus on the possibility that these codes may be vulnerable to malicious manipulation, where an attacker, by modifying just a few pixels in the image of a QR code, could inadvertently change its original message. This motivation is based on the premise that the integrity of the information contained in a QR code is essential in many of the applications where they are used, from e-commerce to healthcare and identity authentication.

\section{Error correcting codes}
To store or transmit information securely without losing important messages, it is necessary to have a method for detecting and correcting errors that may occur during the process. All available communication channels have some degree of noise or interference, such as wear and tear on the image of a QR code or a scratch on a CD, and these factors must be considered when transmitting information.

In this chapter, we will go through the mathematical theory behind error correcting codes. It is important to note that all the theoretical results discussed in this chapter are already present in the academic literature. Our objective is simply to compile and explain these results to facilitate understanding, and then to explore the error resilience in Reed-Solomon codes. If you are already familiar with Reed-Solomon codes, you may proceed directly to Chapter \ref{chap:qrstructure}.

\subsection{Basic Concepts}

A message must be encoded in a predefined alphabet. We can denote this alphabet as a finite set $\mathcal{A} = \{a_1, ..., a_q\}$. Given the alphabet $\mathcal{A}$, we denote by $\mathcal{A}^n$ the set of words of length $n$, and by $\mathcal{A}^* = \bigcupdot_{n \in \mathbb{N}} \mathcal{A}^n$, the set of words from $\mathcal{A}$.

With these concepts, we can define a code $\mathcal{C}$ as a subset of $\mathcal{A}^*$, and a block code $\mathcal{C}_b$ as a subset of $\mathcal{A}^n$, for some $n \in \mathbb{N}$. The elements of a code $\mathcal{C}$ are called the codewords of $\mathcal{C}$.

\begin{example}\phantom{}
\begin{enumerate}
\item Let $\mathcal{A} = \{0, 1\}$. We define $\mathcal{A}^3$ as the binary numbers of 3 bits. We consider $\mathcal{C}_b \subseteq \mathcal{A}^3$, constructed as $\mathcal{C}_b = \{(x_1, x_2, x_3) \:|\: x_1 + x_2 + x_3 \equiv 0 \text{ (mod 2)}\}$. The subset $\mathcal{C}_b$ is a block code.
\item Let $\mathcal{A} = $ the Spanish alphabet, $q = 27$, $\mathcal{C} \subseteq \mathcal{A}^*$, $\mathcal{C} = $ words in Spanish. The subset $\mathcal{C}$ is a code, but not a block code.
\end{enumerate}

\end{example}\medskip

From this point onward, the term `code' will specifically refer to `block code', as these will be of particular interest.

The case of the alphabet $\mathcal{A} = \mathbb{F}_q$, with $q = p^m$, $p$ prime and $m \in \mathbb{N}$, is of special interest since tools from linear algebra and number theory can be applied to design codes with desirable properties, which we will discuss later.
\medskip
\begin{definition}\textbf{Minimum Distance}

Let $\mathcal{A} = \mathbb{F}_q$ be a finite field, with $q = p^m$, $p$ prime and $m \in \mathbb{N}$. The \textit{Hamming distance} between two vectors $u, v \in \mathbb{F}_q^n$ is the number of positions in which the vectors differ, and it can be expressed as the cardinality of the set $\{i \in {1, ..., n} \:|\: u(i) \neq v(i)\}$:
\begin{align*}
d(u, v) = |\{i \in {1, ..., n} \:|\: u(i) \neq v(i)\}|.
\end{align*}
The \textit{minimum distance} of a code $\mathcal{C} \subseteq \mathbb{F}_q^n$ is:
\begin{align*}
d(\mathcal{C}) = \min\{d(u, v) \:|\: u, v \in \mathcal{C}, u \neq v\}.
\end{align*}
\end{definition}\medskip

\begin{definition}\textbf{$[n, M, d]_q$-code.}

We say that $\mathcal{C}$ is an $[n, M, d]_q$-code if $\mathcal{C} \subseteq \mathcal{A}^n$, $q = |\mathcal{A}|$, $M = |\mathcal{C}|$, and with minimum distance $d(\mathcal{C}) = d$.
\end{definition}\medskip

\begin{definition}\textbf{Linear Code}

We say that $\mathcal{C}$ is a linear code if it satisfies the following properties:
\begin{enumerate}
    \item $\mathcal{A} = \mathbb{F}_q$, with $q = p^m$, $p$ prime and $m \in \mathbb{N}$
    \item $\mathcal{C}$ is a subspace vector of $\mathbb{F}_q^n$.
\end{enumerate}

In this case, $\mathcal{C}$ is defined as an $[n, k, d]_q$-code, where $k$ is such that $q^k = |\mathcal{C}|$ and with minimum distance $d(\mathcal{C}) = d$.
\end{definition}\medskip

\begin{definition}\textbf{Weight of a Code}

For $v \in \mathbb{F}_q^n$, we define the weight of $v$ as $w(v) = d(v, (0, ..., 0)) = $ the number of nonzero digits.

We define the weight of the code $\mathcal{C} \subseteq \mathbb{F}_q^n$ as $w(\mathcal{C}) = \min\{ w(v) \:|\: v \in \mathcal{C}, v \neq 0\}$.
\end{definition}\medskip

In the case of linear codes, there is a clear relationship between the Hamming distance and the weight of a code.

\begin{theorem}
If $\mathcal{C} \subseteq \mathbb{F}_q^n$ is a linear code, then $d(\mathcal{C}) = w(\mathcal{C})$.

\end{theorem}
\begin{proof}
If $u, v \in \mathcal{C}$, then $d(u, v) = w(u - v)$, since $u - v$ has nonzero digits exactly at the positions where $u$ and $v$ differ. Moreover, $u - v$ covers the entire code $\mathcal{C}$ when $u$ and $v$ traverse $\mathcal{C}$, as $\mathcal{C}$ is a subspace vector of $\mathbb{F}_q^n$.

Therefore, $d(\mathcal{C}) = \min\{ d(x, y) \:|\: x, y \in \mathcal{C}, x \neq y \} = \min\{ w(x) \:|\: x \in \mathcal{C}, x \neq 0 \} = w(\mathcal{C})$. \qedhere
\end{proof}\medskip

A good $[n, k, d]_q$-code $\mathcal{C}$ is one that can represent a lot of information with short words while being capable of detecting and correcting multiple errors due to its high minimum distance between words. This means it has the following characteristics:
\begin{itemize}
\item \textbf{Small $n$}: The total number of symbols in the words is preferably small.
\item \textbf{Large $k$}: The number of valid words is large.
\item \textbf{Large $d$}: The minimum distance between words is large. This indicates that the code can correct a higher number of errors in the data.
\end{itemize}

\medskip
\begin{example}\label{example}The following is an example of an algorithm for error detection. Let $\mathcal{C}$ be an arbitrary $[n, M, d]_q$-code.

The message $c_1 \in \mathcal{C}$ is transmitted, and $c_2 \in \mathbb{F}_q^n$ is received:
\begin{itemize}[topsep=1ex]
\item If $c_2 \notin \mathcal{C}$, request retransmission of the message.
\item If $c_2 \in \mathcal{C}$, accept $c_2$ as the valid message.
\end{itemize}
If $c_2 \neq c_1$, and $c_2 \in \mathcal{C}$, at least $d$ errors must have occurred. Therefore, we can detect at least $d - 1$ errors. However, this algorithm does not correct any errors.
\end{example}\medskip

Next, we will see a result that relates the existence of an algorithm that corrects and detects errors with the minimum distance of a code.

\renewcommand{\labelenumi}{$\theenumi$.}
\begin{theorem}\label{theorem:algoritmo}
Let $\mathcal{C}$ be an $[n, M, d]_q$-code. Then, the following are equivalent:
\begin{enumerate}
    \item The minimum distance is such that $d \geq 2t + s + 1$.
    \item There exists an algorithm that detects $t+s$ errors and corrects $t$ errors.
\end{enumerate}
\end{theorem}
\begin{proof}\phantom{}

\boxed{1 \Rightarrow 2} 
Assume $d \geq 2t+s+1$. Let's create an algorithm that corrects $t$ errors and detects $t+s$ errors.
\begin{itemize}
    \item Send $c_1\in\mathcal{C}$, and receive $c_2\in\mathbb{F}_q^n$.
    \item Let $c_3\in\mathcal{C}$ be the word closest to $c_2$: $d(c_2, c_3) = \min_{x\in\mathcal{C}} d(c_2, x)$.
    \item If $d(c_2, c_3) \leq t$, accept that $c_3$ is the correct word.
    \item If $d(c_2, c_3) > t$, request retransmission.
\end{itemize}
In this algorithm, the number of errors will be $d(c_1, c_2)$.
\begin{enumerate}
    \item If $d(c_1, c_2) \leq t$, then $d(c_1, c_3) \leq d(c_1, c_2) + d(c_2, c_3) \leq t + t = 2t < d$ $\Rightarrow$ $c_1 = c_3$ and we have correctly corrected the errors.
    \item If $t < d(c_1, c_2) \leq t+s$, then $d(c_1, c_3) > t$, because otherwise $d(c_1, c_3) \leq d(c_1, c_2) + d(c_2, c_3) \leq (t+s) + t = 2t + s < d$ $\Rightarrow$ $c_1 = c_3$. But this leads to a contradiction, since $d(c_2, c_3) = d(c_2, c_1) > t$ $\Rightarrow$ We have detected that there have been up to $t+s$ errors.
\end{enumerate}

\boxed{2 \Rightarrow 1} Suppose there is an algorithm that corrects $t$ errors and detects $t+s$ errors. Assume that $c_1 \in \mathcal{C}$ is sent and $c_2 \in \mathbb{F}_q^n$ is received with $d(c_1, c_2) = s+t$, $d(c_2, c_3) = d-(s+t)$. Let $c_3 \in \mathcal{C}$ be such that $d(c_1, c_3) = d(\mathcal{C}) = d$. If $d \leq 2t+s$, then $d(c_2, c_3) = d - (s+t) \leq 2t+s- (s+t) \leq t$. Then the algorithm would tell us that $c_3$ is the correct word, which is a contradiction, and therefore, $d \geq 2t+s+1$. \qedhere

\end{proof}\medskip

It is noted that the minimum distance in Example \ref{example} satisfies $d \geq (d - 1) + 1$, and the algorithm corrects 0 errors and detects $d - 1$ errors.

\begin{example}\phantom{}

The NIF (Spanish Tax Identification Number) has the structure $x_7x_6x_5x_4x_3x_2x_1x_0\text{-}x_r$ with $x_i \in \{0, ..., 9\}$ and $x_r$ a letter corresponding to the control digit $r$, defined as $r \equiv \sum\limits_{i=0}^{7}10^i x_i$ mod 23. 
The letter corresponding to each digit is associated according to the the table \ref{table:nif}:
\begin{table}[H]
\centering
\caption{Association of the control digit to letters.}
\setlength{\tabcolsep}{7pt}
\renewcommand{\arraystretch}{1.3}
\begin{tabular}{cccc}
\begin{tabular}[t]{|c|c|}
\hline
\textbf{r} & \textbf{Letter} \\ \hline
0  & T \\ \hline
1  & R \\ \hline
2  & W \\ \hline
3  & A \\ \hline
4  & G \\ \hline
5  & M \\ \hline
\end{tabular} &

\begin{tabular}[t]{|c|c|}
\hline
\textbf{r} & \textbf{Letter} \\ \hline
6  & Y \\ \hline
7  & F \\ \hline
8  & P \\ \hline
9  & D \\ \hline
10 & X \\ \hline
11 & B \\ \hline
\end{tabular} &

\begin{tabular}[t]{|c|c|}
\hline
\textbf{r} & \textbf{Letter} \\ \hline
12 & N \\ \hline
13 & J \\ \hline
14 & Z \\ \hline
15 & S \\ \hline
16 & Q \\ \hline
17 & V \\ \hline
\end{tabular} &

\begin{tabular}[t]{|c|c|}
\hline
\textbf{r} & \textbf{Letter} \\ \hline
18 & H \\ \hline
19 & L \\ \hline
20 & C \\ \hline
21 & K \\ \hline
22 & E \\ \hline
\end{tabular} \\
\end{tabular}
\label{table:nif}
\end{table}

\noindent For example:
\begin{enumerate}
    \item The NIF $51234511\text{-}x_r$ would have the control digit $r = 51234511 \text{ mod 23} = 10$, which corresponds to the letter X. Therefore, the complete NIE is $51234511\text{-}X$.
    \item Suppose we receive the NIF $18279322\text{-}G$. To verify that a NIF is correct, we must see that the control digit condition is met. We see that $r = 18279322 \text{ mod 23} = 3$, which corresponds to the letter A. Since the received letter is G, we detect an error.
\end{enumerate}

In this way, the structure of the NIF is a code $\mathcal{C}$ whose alphabet is $\mathcal{A} = \{ (x_7x_6x_5x_4x_3x_2x_1x_0\text{-}l)\:|\: x_i \in \{0, ..., 9\}, l \in \{A, ..., Z\}\backslash\{I, O, U\} \}$.
We see that the minimum distance of this code is 2. It is clear that the distance is greater than 1, as each NIF has a unique associated letter. Furthermore, we can find two NIFs that are at a distance of 2, such as $12345678\text{-}Z$ and $123456789\text{-}S$.

Therefore, applying Theorem \ref{theorem:algoritmo}, we obtain that there exists an algorithm that can detect 1 error but without correcting any.

\end{example}\medskip

If in an $[n, M, d]_q$-code $d$ is increased, $M$ decreases, obtaining fewer possible keys, or $n$ increases, enlarging the size of the words. This situation is known as the Singleton bound.

\begin{theorem}\label{theorem:singleton}
If $\mathcal{C}$ is an $[n, M, d]_q$-code, then $M \leq q^{n-d+1}$. 

Note: This inequality is known as the Singleton bound.
\end{theorem}
\begin{proof}
    Induction on $d$.
    \begin{itemize}
        \item Case $d=1$: It is trivial, since $\mathbb{F}_q^n$ is an $[n, M, 1]_q$-code with $M = q^n$ and any $\mathcal{C} \subseteq \mathbb{F}^n_q$ has $M \leq q^n$.
        \item Suppose the bound is true for codes with distance $d-1$ and let $\mathcal{C}$ be an $[n, M, d]_q$-code. We shorten $\mathcal{C}$ in one coordinate, that is, we construct $\mathcal{C}^*$ as an $[n-1, M, d-1]_q$-code by removing the last coordinate from each key of $\mathcal{C}$. Note that $\mathcal{C}^*$ necessarily has the same number of keys, since the minimum distance is greater than 1, so by removing one coordinate, no pair of keys coincide. By the induction hypothesis, we have that $M \leq q^{(n-1) - (d-1) + 1} = q^{n-d+1}$. \qedhere
    \end{itemize}
\end{proof}\medskip

The linear $[n, k, d]_q$-codes that satisfy the equality in the Singleton bound are called MDS codes (\textit{Maximum Distance Separable}). These codes maximize $d$ by fixing $n$ and $k$. Reed-Solomon codes are MDS linear codes, as we will see later.

\subsection{Reed-Solomon codes}
To describe Reed-Solomon codes, we need to define some additional concepts that form the basis for constructing these codes.

\medskip
\begin{definition}\textbf{Cyclic Codes}

Given $c = (c_0, c_1, ..., c_{n-1})$, we say that $c' = (c_{n-1}, c_0, c_1, ..., c_{n-2})$ is a cyclic shift.

For a linear code $\mathcal{C}$, we say it is a cyclic code if for any $c \in \mathcal{C}$, the cyclic shift $c' \in \mathcal{C}$.
\end{definition}\medskip

\begin{definition}\textbf{Generator and Parity Matrices}

A generator matrix $G$ of a linear code $\mathcal{C}$ is a matrix whose rows form a basis for $\mathcal{C}$, meaning every element of $\mathcal{C}$ can be constructed as a linear combination of the rows of $G$.

If $\mathcal{C}$ is an $[n, k, d]_q$ linear code, its generator matrix is a $k \times n$ matrix:
$$
G = \begin{pmatrix}
  v_{01} & v_{02} & ... & v_{0n}\\ 
  v_{11} & v_{12} & ... & v_{1n}\\ 
  ... & ... & ... & ...\\ 
  v_{k1} & v_{k2} & ... & v_{kn}\\ 
\end{pmatrix}.
$$

A matrix $H$ of dimension $(n-k) \times n$ that satisfies $HG^T = 0$ (or equivalently $GH^T = 0$) is called the parity-check matrix of a linear code $\mathcal{C}$. The columns of this matrix generate the orthogonal complement of $\mathcal{C}$, i.e., an element $w \in \mathcal{A}^n$ is a codeword of $\mathcal{C}$ if $Hw = 0$.
\end{definition}\medskip

For the process of encoding and decoding messages, these matrices are very useful.
On one hand, a generator matrix of a linear code $\mathcal{C}$ helps to generate all the codewords of the code. On the other hand, the parity-check matrix helps to verify that an element of $\mathcal{A}^n$ is a codeword of the code.
\medskip
\begin{example}\label{example:generator} Let $\mathcal{C}$ be a (linear) binary code with generator matrix
$$G = \begin{pmatrix}
  1 & 0 & 1 & 0\\ 
  0 & 1 & 0 & 1
\end{pmatrix}$$
and parity-check matrix $H = G$.

The elements of $\mathcal{C}$ are ${(0, 0, 0, 0), (1, 0, 1, 0), (0, 1, 0, 1), (1, 1, 1, 1)}$. The code is cyclic. Note that $(1, 0, 1, 0)$ is the cyclic shift of $(0, 1, 0, 1)$, and the other cases are trivial.
\end{example}\medskip

\medskip

\begin{definition}\textbf{Finite Field}

Let $f(x)$ be an irreducible polynomial of degree $n \geq 1$.
We define the finite field \[F := \mathbb{F}_q = GF(p^n) = \frac{\mathbb{Z}_p[x]}{f(x)} = \{a_0 + a_1\alpha + ... + a_{n-1}\alpha^{n-1} \:|\:a_i \in \mathbb{Z}_p, f(\alpha) = 0 \}\]
where $GF(p^n)$ is the Galois field of $p^n$ elements, $p$ is prime, and $q = p^n$. Moreover, $F^* = F \backslash \{0\}$ is a cyclic group of order $q-1$ (see \cite{CuerpoFinito}). 

Note: If $\alpha$ is a generator of $F^*$, then $f(x)$ is called a primitive polynomial.
    
\end{definition}\medskip

We define the quotient ring $F_n[x] = \frac{F[x]}{x^n - 1}$, whose elements are the equivalence classes defined by the relation:
\[
f(x) \sim g(x) \iff f(x) \equiv g(x) (\text{mod }(x^n - 1)).
\]
By the Correspondence Theorem (see \cite{szeto2009correspondence}) we know that $F^n$ is isomorphic to $F_n[x]$:
\begin{equation}
\begin{aligned}
  \pi \colon F^n & \longrightarrow F_n[x] \\
  (0, 0, ..., 0) & \longmapsto 0 \\
  (1, 0, ..., 0) & \longmapsto 1 \\
  (0, 1, ..., 0) & \longmapsto x \\
  (0, 0, ..., 1) & \longmapsto x^{n-1}
\end{aligned}
\label{eqn:pi}
\end{equation}

Through the isomorphism \eqref{eqn:pi} we can identify the codewords of a code with polynomials. Note that in a word, performing a cyclic shift of one position is equivalent to multiplying by $x$ in a polynomial:

\begin{equation*}
\xymatrix{
   (c_0, c_1, ..., c_{n-1}) \ar[r]^{\text{\it desp. cicl.}} \ar@<-2pt>[d]_{\pi} & (c_{n-1}, c_0, ..., c_{n-2}) \\
   c_0 + c_1x + ... + c_{n-1}x^{n-1} \ar[r]^{\cdot x} & c_{n-1} + c_0x + ... + c_{n-2}x^{n-1} \ar@<-2pt>[u]_{\pi^{-1}}
}
\end{equation*}

\bigskip

From now on, we will equivalently use both the vector representation and the polynomial representation of elements of a code $\mathcal{C}$. That is, we will say that $f(x) = a_0 + a_1x + \ldots + a_nx^n \in \mathcal{C}$ if $(a_0, a_1, \ldots, a_n) \in \mathcal{C}$.

Next, we will see two examples to familiarize ourselves with the use of the isomorphism \eqref{eqn:pi}.

\begin{example}
Following Example \ref{example:generator}, we have $n=4$. The word $(1, 0, 1, 0)$ would be the polynomial $g(x) = 1 + x^2$. Applying a cyclic shift via the polynomial, we get $xg(x) = x + x^3$, which corresponds to the codeword $(0, 1, 0, 1)$. Furthermore, $x^2g(x) = x^2 + x^4 = 1 + x^2$. The cases $(0, 0, 0, 0)$ and $(1, 1, 1, 1)$ are trivial.    
\end{example}\medskip

\bigskip

\begin{example}
$F = \mathbb{F}_{16} = GF(2^4) = \frac{\mathbb{Z}_2[x]}{f(x)}$, with $f(x) = 1 + x + x^4$.

Suppose we want to calculate $(1100)(1010)$. 
To do this, we need the correspondence table from $F^n$ to $F^n[x]$, where $\beta = (0 1 0 0)$, that is, $x \text{ mod }f(x)$:

\begin{table}[H]
\centering
\setlength{\tabcolsep}{7pt}
\renewcommand{\arraystretch}{1.3}
\caption{Construction of $GF(2^4)$ with $f(x) = 1+x+x^4$.}
\begin{tabular}{L|C|C}
\text{codeword}               & x^i \text{ mod }f(x)              & \text{power of }\beta    \\ \hline
0000                       & 0                                 & -                           \\
1000                       & 1                                 & \beta^0                     \\
0100                       & x                                 & \beta                       \\
0010                       & x^2                               & \beta^2                     \\
0001                       & x^3                               & \beta^3                     \\
1100                       & x^4 \equiv 1+x                    & \beta^4                     \\
0110                       & x+x^2                             & \beta^5                     \\
0011                       & x^2 + x^3                         & \beta^6                     \\
1100                       & x^3 + x^4 \equiv 1 + x + x^3      & \beta^7                     \\
1010                       & x + x^2 + x^4 \equiv 1 + x^2      & \beta^8                     \\
0101                       & x + x^3                           & \beta^9                     \\
1110                       & x^2 + x^4 \equiv 1+x+x^2          & \beta^{10}                  \\
0111                       & x + x^2 + x^3                     & \beta^{11}                  \\
1111                       & x^2 + x^3 +x^4 \equiv 1+x+x^2+x^3 & \beta^{12}                  \\
1011                       & x+x^2+x^3+x^4 \equiv 1+x^2+x^3    & \beta^{13}                  \\
1001                       & x + x^3 + x^4 \equiv 1 + x^3      & \beta^{14}                  \\
1000                       & x + x^4 \equiv 1                  & \beta^{15} = \beta^0
\end{tabular}
\label{table:gf16}
\end{table}
\vspace*{-\baselineskip}

We see that $(1100)(1010) = \beta^4\beta^8 = \beta^{12} = (1111)$.
 
\end{example}\medskip

\subsubsection{Cyclic Codes}
In this section, we will see four results that will help us understand cyclic codes over polynomial rings and how they can be generated by monic polynomials of minimal degree. Additionally, these results provide tools for encoding and analyzing cyclic codes, which will be useful later for developing BCH and Reed-Solomon codes.

\begin{theorem}\label{theorem:cyclicideal}
Let $\mathcal{C}\subseteq F^n$ be a linear code and $\pi:F^n\longrightarrow F_n[x]$ be the isomorphism \eqref{eqn:pi}.
Then, $\mathcal{C}$ is a cyclic code if and only if $\pi(\mathcal{C})$ is an ideal of $F_n[x]$.
\end{theorem}
\begin{proof}\phantom{}

\boxed{\Rightarrow} To see that $\pi(\mathcal{C})$ is an ideal of $F_n[x]$, we need to see that $\pi(\mathcal{C})$ is an additive subgroup of $F_n[x]$ and that the product of an element of $\pi(\mathcal{C})$ by an element of $F_n[x]$ belongs to $\pi(\mathcal{C})$.
\vspace{-\baselineskip}
\begin{enumerate}
\item If $a = \sum\limits_{i=0}^{n-1}a_ix^i$, $b = \sum\limits_{i=0}^{n-1}b_ix^i \in \pi(\mathcal{C})$, then $\textbf{a} = (a_0, \ldots, a_{n-1})$, $\textbf{b} = (b_0, \ldots, b_{n-1}) \in \mathcal{C}$, and therefore $\textbf{a} \pm \textbf{b} \in \mathcal{C}$ and $a \pm b \in \pi(\mathcal{C})$.

\item If $r =  \sum\limits_{i=0}^{n-1}r_ix^i \in F_n[x]$, then $r \cdot a = \sum\limits_{i=0}^{n-1}r_ix^ia$, where $x^ia$ is a cyclic shift so $x^ia \in \pi(\mathcal{C})$. Since $r_i \in F$ and $\pi(\mathcal{C})$ is an $F$-vector space, we have $r_ix^ia \in \pi(\mathcal{C})$ for each $i = 0, \ldots, n-1$, thus $ra \in \pi(\mathcal{C})$.
\end{enumerate}

\boxed{\Leftarrow} Let $\alpha, \beta \in F$ and $\textbf{a}, \textbf{b} \in \mathcal{C}$. By hypothesis, $\alpha\pi(\textbf{a}) + \beta\pi(\textbf{b}) \in \pi(\mathcal{C})$. By linearity of $\pi$, we have $\alpha\pi(\textbf{a}) + \beta\pi(\textbf{b}) = \pi(\alpha\textbf{a} + \beta\textbf{b})$, thus $\alpha\textbf{a} + \beta\textbf{b} \in \mathcal{C}$ and $\mathcal{C}$ is a linear code. Additionally, if $\pi(a) \in \pi(\mathcal{C})$, then $x\cdot \pi(a) \in \pi(\mathcal{C})$, so the cyclic shift of $a$ is part of the code and $\mathcal{C}$ is cyclic. \qedhere
\end{proof}\medskip

\medskip

\begin{proposition}\label{prop:cyclicmonic}
    Let $\mathcal{C} \subseteq F^n$ be a cyclic code. Then there exists a unique monic polynomial $g(x)$ of minimal degree such that
    \[
    \mathcal{C} = \langle g(x) \rangle = \{t(x)g(x) \in F_n[x] \:|\: t(x) \in F[x] \}
    \]
    Furthermore, $g(x) \divides (x^n-1)$.
\end{proposition}
\begin{proof}
    We will first show the existence of $g(x)$. 
    
    Let $S = \{\text{deg}(g(x)) \:|\: g(x) \in \mathcal{C}\} \subseteq \mathbb{N}$. Since $S$ is a non-empty subset of
    $\mathbb{N}$, there exists a smallest element $r \geq 0$ such that:
    \[
    g_1(x) = a_0 + a_1x + \ldots + a_rx^r \in \mathbb{C}, a_r \neq 0
    \]
    Let $g(x) = a_r^{-1}g_1(x) = x^r + \sum_{i=0}^{r-1}a_r^{-1}a_ix^i$. Thus $g(x)$ is a monic polynomial of minimal degree in $\mathcal{C}$.

    Let $\mathcal{C}' = \langle g(x) \rangle = \{t(x)g(x) \:|\: t(x) \in F[x]\}$, the ideal generated by $g(x)$. Since $g(x) \in \mathcal{C}$ and $\mathcal{C}$ is cyclic, for any $t(x) \in F[x]$, we have $t(x)g(x) \in \mathcal{C}$ and $\mathcal{C}' \subseteq \mathcal{C}$.

    If $\mathcal{C} \setminus \mathcal{C}' \neq \emptyset$, then there exists $h(x) \in \mathcal{C} \setminus \mathcal{C}'$. Using the division algorithm for $F[x]$, there exist $q(x), r(x) \in F[x]$ such that
    \[
    h(x) = q(x)g(x) + r(x), \text{ with deg}(r(x)) < \text{deg}(g(x))
    \]
    However, $h(x) \in \mathcal{C}$, and $\mathcal{C}' \subseteq \mathcal{C}$, thus $q(x)g(x) \in \mathcal{C}$. Since $\mathcal{C}$ is a vector space, $h(x) - q(x)g(x) = r(x) \in \mathcal{C}$. Since $r(x) \in \mathcal{C}$ and $r(x) \neq 0$, $\text{deg}(r(x))$ contradicts the minimality of $\text{deg}(g(x))$. Hence, $\mathcal{C} = \mathcal{C}' = \langle g(x) \rangle$.

    Finally, since $\mathcal{C}$ is a linear code, it has $q(x) = \frac{x^n-1}{g(x)}$ elements, and all polynomials in $\mathcal{C}$ divide $x^n-1$.
\end{proof}\medskip

\medskip

\begin{proposition}\label{prop:dimensionciclic}
    The set $\{g(x), xg(x), ..., x^{n-r+1}g(x)\}$ forms a basis for $C$, and therefore, $C$ has dimension $n - r$.
\end{proposition}
\begin{proof}
    Let's assume $\text{deg}(g(x)) = r$. For this proof, we will first show that $\mathcal{C} = \{u(x)g(x) \:|\: \text{deg}(u(x)) < n-r\}$.
    
    From the previous proposition, we know that:
    \[
    \mathcal{C} = \langle g(x) \rangle = \{g(x)f(x) \:|\:f(x) \in F_n[x]\}
    \]
    First, we'll show that we can restrict $f(x)$ to have $\text{deg}(f(x)) < n-r$. We know that $x^n-1 = h(x)g(x)$ for some polynomial $h(x)$ of degree $n-r$. By dividing in $F[x]$, we get $f(x) = q(x)h(x) + u(x)$, where $\text{deg}(u(x)) < n-r$ or $u(x) \equiv 0$. Thus,
    \begin{align*}  
    &f(x)g(x) = q(x)h(x)g(x) + u(x)g(x) = q(x)(x^n-1) + u(x)g(x) \Rightarrow \\
    \Rightarrow &f(x)g(x) = u(x)g(x)\text{ in }F_n[x]
    \end{align*}
    Therefore, $\text{deg}(f(x)) < n-r$.

    Since the polynomials $\{1, x, ..., x^{n-r+1}\}$ are $n-r$ linearly independent polynomials of degree less than $n-r$, we have that $\{g(x), xg(x), ..., x^{n-r+1}g(x)\}$ forms a basis for $\mathcal{C}$ and $\dim(\mathcal{C}) = n-r$. \qedhere
\end{proof}\medskip

\bigskip

\begin{theorem}
    Let $g(x) = q_1(x)q_2(x)...q_t(x)$ be a product of irreducible factors of $x^n-1$ over $F$, and let $\alpha_i$ be roots of $q_i(x)$. Then:
    \[
    C = \langle g(x) \rangle = \{f(x) \in F_n[x] \:|\: f(\alpha_1) = 0, ... f(\alpha_t) = 0\}
    \]
\end{theorem}
\begin{proof}
    Let $C = \langle g(x) \rangle$ and $S_{\alpha_1, ..., \alpha_t} = \{f(x) \in F_n[x] \:|\: f(\alpha_1) = 0, ... f(\alpha_t) = 0\}$
    
    \boxed{\subseteq} Let $f(x) \in C$. Clearly, $f(x) = a(x)g(x)$. Thus $f(\alpha_i) = 0$ for $1 \leq i \leq t$, hence $f(x) \in S_{\alpha_1, ..., \alpha_t}$ and $C \subseteq S_{\alpha_1, ..., \alpha_t}$.

    \boxed{\supseteq} Let $f(x) \in S_{\alpha_1, ..., \alpha_t}$. This means $f(\alpha_i) = 0$. Since $q_i(x)$ is irreducible and monic and $q_i(\alpha_i) = 0$, $q_i(x)$ divides $f(x)$ for $1 \leq i \leq t$. This implies:
    \[
    f(x) \in \bigcap\limits_{i=1}^{t} \langle q_i(x) \rangle = \langle \text{lcm}\{q_1(x), ..., q_t(x)\} \rangle
    \]

    Since $q_i$, $i = 1, ..., t$, are irreducible factors over $F$ and do not share common roots with $q_j$, $j \neq i$, it follows that:
    \[
    f(x) \in \langle \text{lcm}\{q_1(x), ..., q_t(x)\} \rangle = \langle \prod\limits_{i=1}^{t}q_i(x) \rangle = \langle g(x) \rangle
    \]

    Therefore, $S_{\alpha_1, ..., \alpha_t} \subseteq C$, and we have equality. \qedhere
    
\end{proof}\medskip

\subsubsection{Detailed Functioning}
Bose-Chaudhuri-Hocquenghem (BCH) cyclic codes are an important class of error-correcting codes in coding theory. These codes are particularly known for their ability to correct multiple errors and their robustness in environments where data transmission reliability is crucial.

BCH codes are considered precursors to Reed-Solomon codes due to their relationship and structural similarities. In fact, Reed-Solomon codes can be seen as a specialized subtype of BCH codes (see \cite{wicker1994introduction}).

\begin{definition}\textbf{BCH Codes}\label{def:bch}

    Let $g(x) = q_1(x)q_2(x)...q_t(x)$ be a product of irreducible factors of $x^n-1$ over $F$, and let $\alpha_i$ be roots of $q_i(x)$. Then:
    \[
    C = \langle g(x) \rangle = \{f(x) \in F_n[x] \:|\: f(\alpha_1) = 0, ... f(\alpha_t) = 0\}
    \]
\end{definition}\medskip

\begin{example}
    Consider the BCH code with $q=2$ (i.e., $p=2$ and $r=1$), $m=3$, and hence $n=2^3 - 1$. Take an expected minimum distance $\delta = 3$. We need to consider a cyclic code whose generator polynomial has roots $\alpha$ and $\alpha^2$, where $\alpha \in \mathbb{F}_8$ is a primitive $n$-th root of unity.

    The factorization of $x^7 - 1$ into irreducible factors is:
    \[
    x^7 - 1 = (x - 1)(x^3 + x^2 + 1)(x^3 + x + 1)
    \]
    Let $f(x) = x^3 + x + 1$. This gives $f(\alpha) = 0$, hence $\alpha^3 = \alpha + 1$. Also, $f(\alpha^2) = (\alpha^2)^3 + \alpha^2 + 1 = (\alpha^2 + 1) + \alpha^2 + 1 = 0$. Thus, if $f(\alpha) = 0$, then $f(\alpha^2) = 0$. With this reasoning, the other root of $f(x)$ is $\alpha^4$. We can rewrite $f(x)$ as
    \[
    f(x) = x^3 + x + 1 = (x-\alpha)(x-\alpha^2)(x-\alpha^4)
    \]
    The other powers of $\alpha$ are to be roots of $(x-1)$ and $(x^3 + x^2 + 1)$. From this, we can build the polynomials $q_i$:
    \begin{align*}  
    q_0(x) &= x-1 \\
    q_1(x) &= q_2(x) = q_4(x) = x^3 + x + 1 \\
    q_3(x) &= q_5(x) = q_6(x) = x^3 + x^2 + 1
    \end{align*}
    Therefore, the generator polynomial of the BCH code is 
    \[g(x) = \text{lcm}\{q_1(x), q_2(x)\} = f(x)\]

    Thus, with $k = n - \text{deg}(g(x)) = 4$, we have constructed a BCH $[7, 4, 3]_2$-code.
\end{example}\medskip

\begin{definition}\textbf{Reed-Solomon code}

Let $q = p^r$, with $p$ prime. A Reed-Solomon code is a BCH $[n, k, d]_q$-code, where $n = q - 1$, $k = n-d+1$, and the alphabet $\mathcal{A} = GF(q)$. 
    
According to Definition \ref{def:bch} of a BCH code, the multiplicative order $m$ of $q$ modulo $n$ is fixed as 1.

These codes are often denoted as RS$(n,k)$.
\end{definition}\medskip

\begin{remark}
Reed-Solomon codes have the advantage that we can construct \( GF(q) \) as an extension of \( GF(p) \), \( GF(q) = GF(p^r) = \frac{GF(p)[x]}{f(x)} \), where \( f(x) \) is an irreducible polynomial of degree \( r \) over \( \mathbb{F}_p \).

The case with \( p = 2 \) provides greater computational simplicity, as elements of \( GF(q) \) can be expressed as binary polynomials with coefficients 0 or 1. Another advantage is that since the unity root \( \alpha \in GF(q^m) = GF(q) \), the minimal polynomial of \( \alpha^i \) is trivially \( (x - \alpha^i) \). This simplifies the calculation of the generator polynomial, as it is the product of these minimal polynomials:
\[
g(x) = \prod_{i=1}^{\delta - 1}(x-\alpha^i)
\]
\end{remark}

It is also observed that Reed-Solomon codes are Maximum Distance Separable (MDS) codes:
\begin{align*} 
k &= n - \text{deg}(g(x)) = n - (\delta - 1) \\
d &= \delta = n - k + 1
\end{align*} 
Thus, Singleton's bound \ref{theorem:singleton} is satisfied.

\subsection{Implementation of Reed-Solomon Codes}
In this section, we explore an example of a Reed-Solomon code along with an algorithm for encoding messages used in QR codes.

\begin{example}\label{ejemp:rs255223}
Let's consider the code  RS$(255, 223)$, where we have \( k = 223 \) and \( n = 255 = q - 1 \). 
We can deduce that \( q=256=2^8 \), meaning that we work over \( GF(2^8) \). Furthermore, since \( q = 2^r \), then \( r = 8 \). That is, the degree of the primitive polynomial \( f(x) \) characterizing the field \( GF(2^8) = \frac{GF(2)}{f(x)} \) is \( r \). We will use \( f(x) = 1+x^2+x^3+x^4+x^8 \) as the primitive polynomial.

The minimum distance is \( d=n-k+1=33 \). By Theorem \ref{theorem:algoritmo}, we know that \( d\geq2t+e \), where \( t \) is the number of errors that can be corrected, and \( e \in \{0, 1\} \). Therefore, \( t = 16 \), so we can correct up to 16 errors with this code.

\end{example}\medskip

\subsubsection{QR Code Workspace}
In the context of generating and decoding QR codes, the workspace refers to the mathematical field over which encoding and decoding operations are performed. These operations are carried out over a finite field, specifically over the Galois field GF(256).

QR codes break down information into bytes, or equivalently, integers between 0 and 255. 

\begin{example}\label{ejemp:algoritmocodif}
Following Example \ref{ejemp:rs255223}, we will demonstrate an algorithm to construct keys as a sequence of polynomial coefficients.

Suppose the original input message is the sequence of symbols:
\[ (m_0, m_1, \dots, m_{222}) \]
where \( m_i \) are the message symbols.
\begin{enumerate}
    \item Convert the message into coefficients of a polynomial \( f(x) \) of degree \( k-1 = 222 \):
    \[ f(x) = \sum_{i=0}^{222} m_i x^i \]

    \item Multiply \( f(x) \) by \( x^{n-k} \) to obtain a polynomial of degree \( n-1 = 254 \):
    \[ f(x) \cdot x^{32} = \sum_{i=0}^{222} m_i x^{i+32} \]

    \item Use a specific generator polynomial \( g(x) \) for RS$(255, 223)$.
    
    \item Divide \( f(x) x^{32} \) by \( g(x) \) to obtain the encoded polynomial \( s(x) \):
    \[ s(x) = (f(x) x^{32}) \mod g(x) \]
\end{enumerate}

This algorithm is used for generating error correction codes that are incorporated into a QR code (see \cite{rscodewordgithub} and \cite{rscodewordwikiversity}).
\end{example}\medskip

The algorithm in Example \ref{ejemp:algoritmocodif} is applied with a reduced input message, assigning 0 to symbols from a fixed index onward. 
For example, to encode 13 symbols using the above algorithm, the input vector would be:
\[ (m_0, m_1, \dots, m_{12}, 0, \dots, 0)\text{.} \]
This reduces the complexity of the encoding and decoding process by decreasing the dimensionality of the polynomial division in the final step of Example \ref{ejemp:algoritmocodif}.

\section{Structure of QR}\label{chap:qrstructure}
A QR code is a two-dimensional representation of data that is resilient to errors. This is achieved thanks to the use of Reed-Solomon codes. 

The process of creating a QR code consists of seven key stages:\label{qrstages}
\begin{enumerate}
    \item \textbf{Data Analysis:} In this stage, the optimal encoding mode is determined for the text string to be converted into a QR code. The encoding modes include:
        \begin{itemize}
            \item Numeric Mode: for encoding numbers.
            \item Alphanumeric Mode: for alphanumeric characters.
            \item Byte Mode: for binary data.
            \item Kanji Mode: for Japanese characters.
        \end{itemize}

    \item \textbf{Data Encoding:} The text string is converted into a sequence of bytes using the selected encoding mode from the previous stage.

    \item \textbf{Error Correction Encoding:} Error correction codewords are generated using error correction codes like Reed-Solomon. These codewords are redundant and are added to the original message to enable error detection and correction during QR code decoding.

    \item \textbf{Final Message Structuring:} Data and error correction codewords are combined and arranged in a specific order to form the final message that will be inserted into the QR code.

    \item \textbf{Module Placement in Matrix:} The resulting message bits are laid out in the module matrix (pixels) that form the QR code pattern. This pattern uses interleaving to enhance error resistance and correction capability, and it changes according to the QR code version.

    \item \textbf{Data Masking:} To improve readability and error detection capability, a masking pattern is applied to the data matrix. This pattern alters bits in the matrix according to a specific masking formula.

    \item \textbf{Format and Version Information:} Finally, format and version information is added to the QR code. Format information includes details about data masking and error correction level, while version information indicates the specific variant of the QR standard used.
\end{enumerate}

The amount of information that can be stored in a QR code depends on its version, error correction level, and the type of data stored (binary, numeric, alphanumeric, and Kanji). In the analysis we will conduct, we will use binary data type, as it is the most generic and can contain any type of information. Therefore, we will skip the data analysis step of the QR creation process.

Below is the amount of bytes that can be stored for some versions, depending on the aforementioned restrictions:

\begin{table}[H]
\centering
\caption{Storage in a QR code depending on its version, size, and error correction level.}
\setlength{\tabcolsep}{7pt}
\renewcommand{\arraystretch}{1.3}
\begin{tabular}{|c|c|c|c|}
\hline
Version & Size & Error Correction Level                                 & Storable Bytes                                           \\ \hline
1       & 21x21  & \begin{tabular}[c]{@{}c@{}}L\\ M\\ Q\\ H\end{tabular} & \begin{tabular}[c]{@{}c@{}}17\\ 14\\ 11\\ 7\end{tabular}  \\ \hline
2       & 25x25  & \begin{tabular}[c]{@{}c@{}}L\\ M\\ Q\\ H\end{tabular} & \begin{tabular}[c]{@{}c@{}}32\\ 26\\ 20\\ 14\end{tabular} \\ \hline
3       & 29x29  & \begin{tabular}[c]{@{}c@{}}L\\ M\\ Q\\ H\end{tabular} & \begin{tabular}[c]{@{}c@{}}53\\ 42\\ 32\\ 24\end{tabular} \\ \hline
\end{tabular}
\label{table:qrcapacity}
\end{table}
\vspace*{-\baselineskip}

In Table \ref{table:qrcapacity}, we observe four different error correction levels:
\begin{itemize}\label{items:porcentajeerror}
    \item \textbf{L (low)}: 7\% of erroneous data can be corrected.
    \item \textbf{M (medium)}: 15\% of erroneous data can be corrected.
    \item \textbf{Q (quartile)}: 25\% of erroneous data can be corrected.
    \item \textbf{H (high)}: 30\% of erroneous data can be corrected.
\end{itemize}
For any error correction level, a QR code reader would be capable of detecting twice the amount of erroneous data, as Reed-Solomon codes are used to both correct and detect errors.

\subsection{Decomposition of a QR Code}\label{section:separacionbloques}

The data represented in a QR code is divided into several blocks to simplify the decoding process, thereby limiting the complexity of the problem. We will further study this complexity later on.

\subsubsection{Image}

Next, let's see some examples of the format of QR codes when rendered in a two-dimensional matrix.

\begin{figure}[H]
  \centering
  \includegraphics[width=0.5\textwidth]{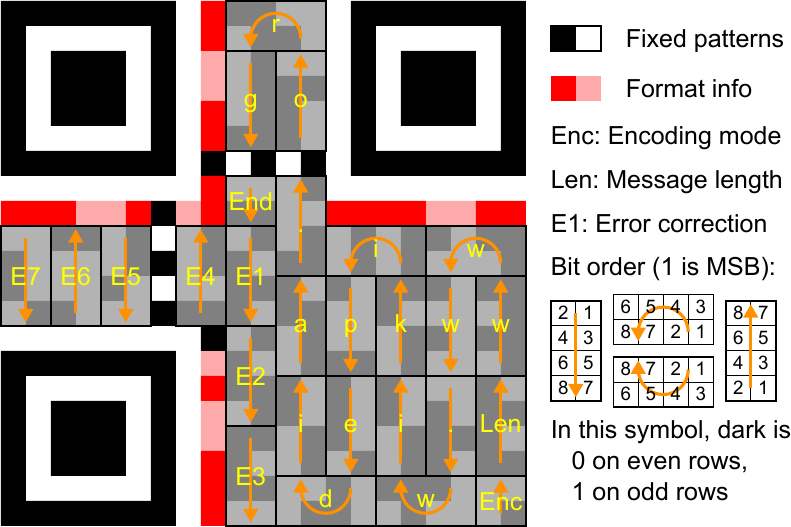}
  \caption{Layout of data in a QR code version 1. Image obtained from \cite{WikiQRCode}.}\label{fig:qrv1}
\end{figure}

QR codes are depicted in a two-dimensional matrix that only takes values 0 and 1, represented in black and white colors, respectively.

In the QR code of Figure \ref{fig:qrv1}, there is a single data block (with elements identified as [Enc, Len, w, w, w, ., w, i, k, i, p, e, d, i, a, ., o, r, g, End]). The number of error correction codes depending on the correction level and version is fixed (see \cite{ThonkyTable}). Since there are seven error correction codes (named [E1, E2, E3, E4, E5, E6, E7]) in this QR code, it implies that level L error correction is being used. A QR reader would identify this information through the format data, represented in red.

In Figure \ref{fig:qrv1}, we also see the order in which information is laid out in the QR code. This order is represented by the flow of arrows, encoding the data block first and then its error correction codes. The bit order within each data element or error correction code is indicated by the four possible ways of traversing an element.

We can also distinguish some control data introduced as QR data, such as the encoding format (Enc element), which consists of four bits, the message length (Len), and a final padding (End) indicating the end of the message. In this specific QR code, we have a total of 17 bytes to represent the desired message, which is ``www.wikipedia.org''.

\begin{figure}[H]
  \centering
  \includegraphics[width=0.35\textwidth]{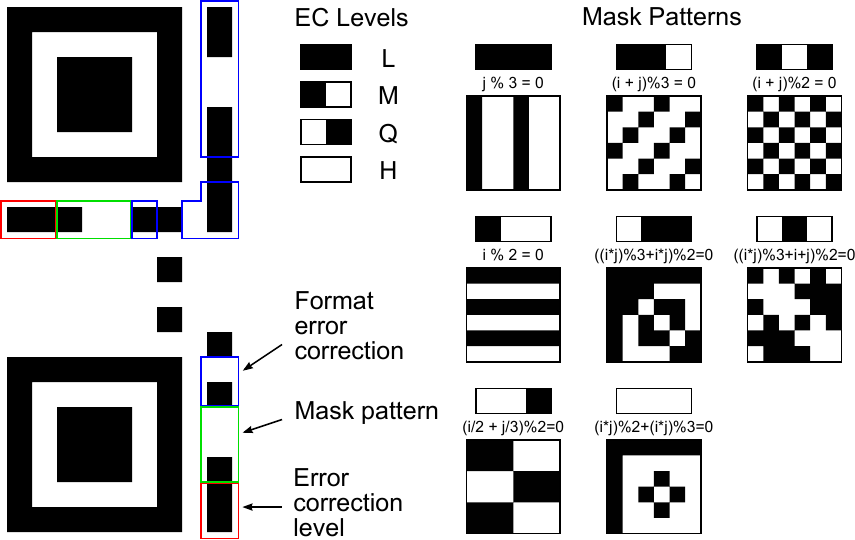}
  \caption{QR format information. Image obtained from \cite{WikiQRCode}.}\label{fig:qrformat}
\end{figure}

\medskip

The above figure shows different components that display the QR format information:
\begin{enumerate}
    \item In red, the correction level used (L, M, Q, or H). This information is represented in two bits.
    \item In green, the mask pattern applied to the QR code. There are eight possible masks, shown on the right side of the image. This information is represented in three bits.
    \item In blue, the error correction codes for the above values. With the five bits comprising the union of the correction level and the mask pattern, ten error correction codes are generated to provide greater resilience against errors.
\end{enumerate}

The mask is a pattern applied to the data arranged in the two-dimensional matrix using an XOR operation to modify the QR code to make it as easy to read as possible. The optimal mask for the represented data is found by testing the seven existing combinations. However, a QR code reader only decodes the data without analyzing whether the used mask is optimal.

\medskip

In the example of Figure \ref{fig:qrv1}, mask 011 is used, which corresponds to performing the XOR operation on the rows of the two-dimensional matrix alternately. Therefore, without considering the mask, we would read the encoding format as 0111. However, by applying the mask again, we cancel its effects. Thus, we obtain that the encoding is 0100, corresponding to binary mode.

\begin{figure}[H]
  \centering
  \includegraphics[width=0.5\textwidth]{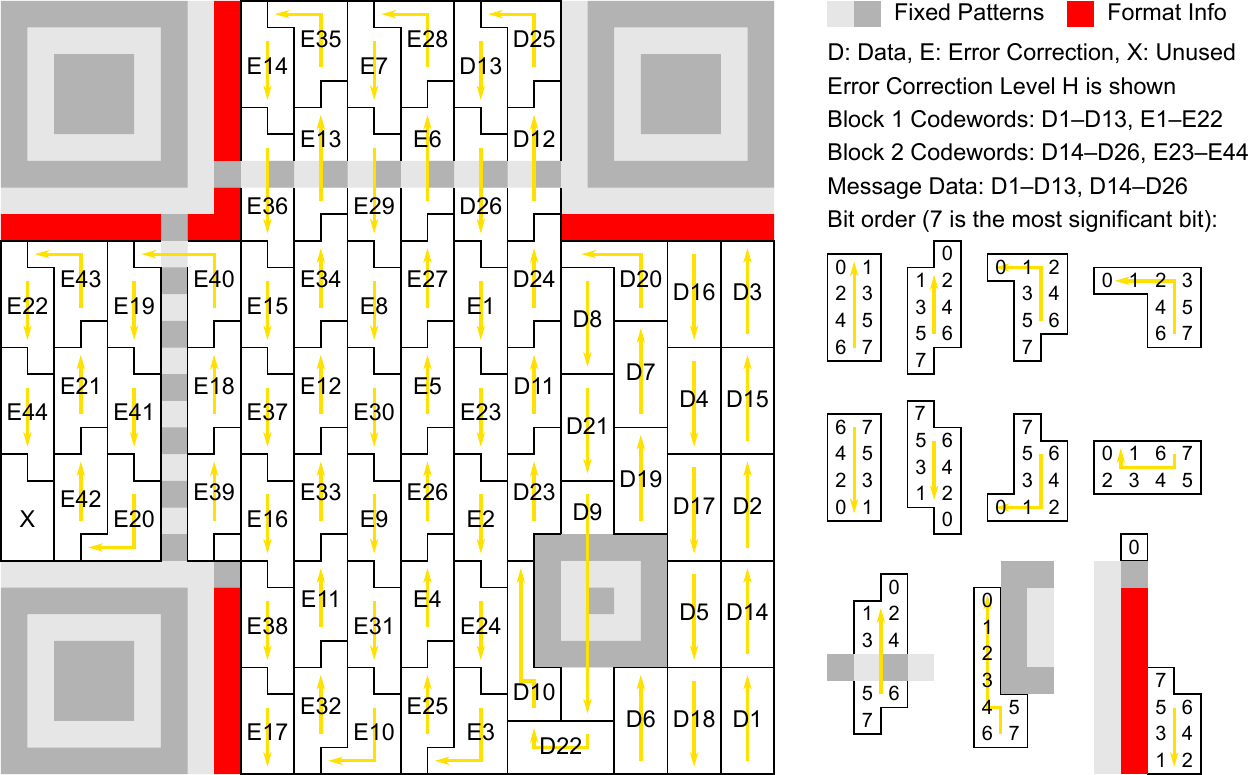}
  \caption{Arrangement of data in a QR code version 3. Image obtained from \cite{WikiQRCode}.}\label{fig:qrv3}
\end{figure}

In the above figure, we observe a more complex case, a QR code version 3, with error correction level H.
The message contains 26 bytes of data, of which 2 bytes are reserved for encoding format, message length, and final padding. Therefore, as indicated in Table \ref{table:qrcapacity}, we have 24 bytes of effective data.

\medskip

Furthermore, in version 3 with error correction level H, the information is divided into two blocks, hence two independent Reed-Solomon codes are used. The blocks are interleaved (the order of data arrangement is D1, D14, D2, D15, ...) to correct errors even when localized damage occurs in any area of the image.

\section{QR selective edition vulnerability}

As specified earlier, a QR code corrects a percentage of errors approximately depending on the error correction level. However, error correction in QR codes focuses on treating each byte as an 8-bit entity, meaning correction is performed at the byte level. This methodology provides acceptable robustness in data recovery under adverse conditions, such as visual defects in parts of the QR. It's important to note, though, that this error correction is not conducted at the bit level, which introduces certain limitations in the precision of error correction.

A vulnerability is observed, stemming from the possibility of making subtle adjustments at the bit level. While a QR code can correct a predefined percentage of errors, there exists a probability of strategically manipulating the code by modifying specific bit values without completely altering the error correction keys, thereby achieving a slightly different message. This has significant implications in commercial environments, especially those utilizing QR codes for facilitating financial transactions.

\subsection{Proof of Concept}\label{section:poc}
We will explore an example where a data modification in a QR affects the change and remains readable by a QR reader. Let's assume the original text string in the QR is ``Id: 1234567''. We use error correction level Q, meaning it can correct up to 25\% of erroneous data. This results in the usage of the code $RS(26,13)$.

\begin{figure}[H]
  \centering
  \includegraphics[width=0.25\textwidth]{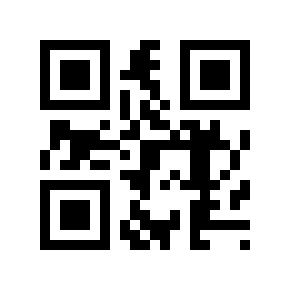}
  \caption{QR code with data ``Id: 1234567'' and error correction level Q.}\label{fig:qrprueba}
\end{figure}

The proposed text consists of 11 characters. Using error correction level Q, according to Table \ref{table:qrcapacity}, the most appropriate version to use is the first one. In this version, a total of 26 bytes of information are displayed, divided as follows:
\begin{itemize}
    \item 11 bytes of data to display from the text.
    \item 13 bytes related to error correction codes.
    \item 2 bytes of information about the encoding format and text length.
\end{itemize}

Using open-source libraries like \textit{qrcode} (see \cite{qrcode}) or \textit{qrcodegen} (see \cite{qrcodegen}), we can obtain both a byte string with the information to represent and the QR representation in an image format. As seen previously in Figure \ref{fig:qrformat}, the image representation of data uses a mask that depends on the data itself. However, since QR code reading does not consider the use of the most optimal mask but tries to decode the data regardless of the mask used, we don't need to deal with the QR image. Therefore, we will work with the data as a byte string.

\bigskip

Suppose we want to modify a character in the original text, changing from ``Id: 1234567'' to ``Id: 1234566''. The QR code in Figure \ref{fig:qrprueba} is generated with the following byte string:
\begin{align*}
b_1 =\text{ }&[64, 180, 150, 67, 162, 3, 19, 35, 51, 67, 83, 99, \mathbf{112}, \\
&\mathbf{196, 144, 22, 34, 115, 74, 89, 202, 212, 234, 197, 39, 150}]
\end{align*}
The byte string associated with the QR of the second text is:
\begin{align*}
b_2 =\text{ }&[64, 180, 150, 67, 162, 3, 19, 35, 51, 67, 83, 99, \mathbf{96}, \\
&\mathbf{188, 116, 128, 47, 172, 71, 62, 26, 14, 96, 156, 143, 69}]
\end{align*}

The first 13 elements represent the numerical encoding of the data. The last 13 elements are the error correction codes for this data. Both in \( b_1 \) and \( b_2 \), the elements that differ between the two byte arrays are highlighted in bold. This change between the byte arrays is due to the modified data byte, changing from ``7'' to ``6'', and the consequent adjustment of the data's error correction codes.

As this QR codes uses $RS(26,13)$, a QR code reader can correct \( \text{floor}(\frac{14-1}{2}) = 6 \) errors (see Theorem \ref{theorem:algoritmo}), because the minimum distance is $d = 26 - 13 + 1 = 14$. Therefore, if we want to modify string \( b_1 \) so that it interprets as string \( b_2 \), at least seven error correction codes must match those of \( b_2 \), in addition to the modified data byte.

At first glance, modifying eight bytes in a QR code might seem costly, as it equates to 64 bit-level modifications or, equivalently, 64 pixel changes in the image. However, if we represent \( b_1 \) and \( b_2 \) as binary numbers, we observe that the number of bits to modify decreases significantly:

\begin{table}[h!]
\centering
\caption{Binary representation of differing bytes of QR with texts ``Id: 1234567'' and ``Id: 1234566''.}
\setlength{\tabcolsep}{7pt}
\renewcommand{\arraystretch}{1.3}
\begin{center}
\begin{tabular}{c|c|c|c|c|c|}
\cline{2-6}
\multicolumn{1}{c|}{Index}           & 11       & 12       & 13       & 14       & 15       \\ \hline
\multicolumn{1}{|c|}{$b_1$}           & 01100011 & 01110000 & 11000100 & 10010000 & 00010110 \\ \hline
\multicolumn{1}{|c|}{$b_2$}           & 01100011 & 01100000 & 10111100 & 01110100 & 10000000 \\ \hline
\multicolumn{1}{|c|}{Bits Modified}     & 0        & 1        & 4        & 4        & 4        \\ \hline
\end{tabular}

\vspace*{2 mm}

\begin{tabular}{c|c|c|c|c|c|}
\cline{2-6}
\multicolumn{1}{c|}{Index}          & 16        & 17       & 18       & 19       & 20      \\ \hline
\multicolumn{1}{|c|}{$b_1$}          & 00100010  & 01110011 & 01001010 & 01011001 & 11001010\\ \hline
\multicolumn{1}{|c|}{$b_2$}          & 00101111  & 10101100 & 01000111 & 00111110 & 00011010\\ \hline
\multicolumn{1}{|c|}{Bits Modified}    & 3         & 7        & 3        & 5        & 3       \\ \hline
\end{tabular}

\vspace*{2 mm}

\begin{tabular}{c|c|c|c|c|c|}
\cline{2-6}
\multicolumn{1}{c|}{Index}          & 21        & 22       & 23       & 24       & 25       \\ \hline
\multicolumn{1}{|c|}{$b_1$}          & 11010100  & 11101010 & 11000101 & 00100111 & 10010110 \\ \hline
\multicolumn{1}{|c|}{$b_2$}          & 00001110  & 01100000 & 10011100 & 10001111 & 01000101 \\ \hline
\multicolumn{1}{|c|}{Bits Modified}    & 5         & 3        & 4        & 3        & 5        \\ \hline
\end{tabular}
\end{center}
\label{table:comparison}
\end{table}

A total of eight bytes need to be modified, so we can choose those requiring the fewest bit changes. In this case, we select bytes with indices $12, 13, 14, 16, 18, 20, 22, 24$. The total number of bit-level modifications to alter the message decreases from 64 to 24.

\smallskip

Furthermore, the total number of encoded bits in this QR code is $26 * 8 = 208$, hence the percentage of modified bits is $11.54\%$. This percentage is significantly lower than the error data rate detectable by a QR code with Q level error correction, which is $50\%$. It is even lower than the error correction capability of such a QR code, which is $25\%$.

\smallskip

This data may seem contradictory to the error correction capabilities advertised by Denso-Wave. However, it's important to note that the error correction rate refers to the number of codewords arranged in the QR code, i.e., the number of bytes. Error correction operates at the byte level, and with the method described earlier, we are modifying 8 out of the total 26 bytes, which means we are altering $30.77\%$ of the data, no longer correctable by a QR code reader.

\begin{figure}[H]
  \centering
  \includegraphics[width=0.8\textwidth]{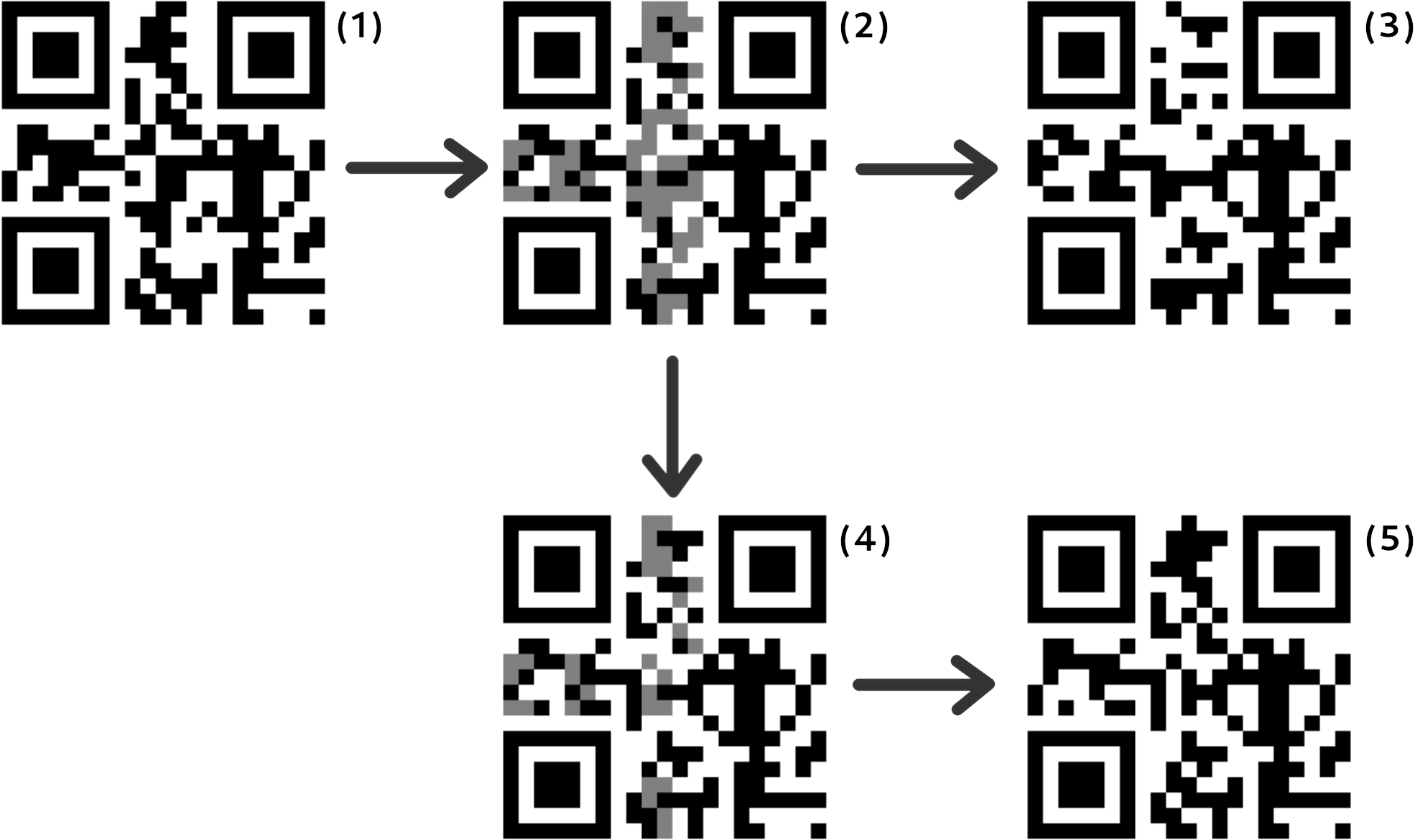}
  \caption{Diagram of the QR code manipulation process.}
  \label{fig:qrproceso}
\end{figure}

In Figure \ref{fig:qrproceso}, we observe the conversion from the original QR code, marked as $(1)$, to the QR code with the last character modified, marked as $(3)$. In the intermediate step $(2)$, all necessary changes to the original QR code are shaded in gray. Using the selective alteration method described earlier, we can modify only the necessary bits to misinterpret the data. Thus, we obtain QR code $(4)$, which is the same as step $(2)$ with selective modifications applied. By reverting the gray pixels to their original state, we obtain QR code $(5)$, providing the same result as QR code $(3)$.

\subsection{Extreme Case}

It may be of interest to compute the QR code closest to another given QR code. In this context, closeness to another QR code refers to minimizing the bit-level changes as explained in Section \ref{section:poc}.

\smallskip

As mentioned in Section \ref{section:separacionbloques}, a QR code encodes information in blocks, each with its own error correction codes. This means that the information contained within a QR code is organized into independent blocks.

\smallskip

When considering changing the information in a QR code by modifying just one character, we can leverage this block structure. Specifically, if a QR code consists of multiple blocks, changing a single character in the code involves making adjustments within only one of those blocks.

To simplify our study and analysis, we focus on the scenario where the QR code consists of a single block. In this scenario, any change in the encoded information will require adjustments within this single block, providing a more direct understanding of how local modifications affect the structure and content of the QR code.

\subsubsection{Specific Example 1}\label{subsection:ejemploparticular}

Let's assume the QR code we want to modify uses error correction level L, which is the lowest, capable of correcting up to 7\% of errors in the QR code. This level of correction allows for 17 bytes of text and generates only 7 error correction codes. Thus, we are considering the code $RS(26,19)$.

Suppose the encoded text in our QR code is ``Id: bhavuksikka'', which is composed of 15 characters (or equivalently bytes). Following the notation used in Section \ref{section:poc}, the initial string is generated by the following byte array:
\begin{align*}
b_1 =\text{ }&[64, 244, 150, 67, 162, 6, 38, 134, 23, 103, 86, 183, 54, 150, 182, 182, 16, 236, 17,\\
&235, 223, 145, 221, 73, 238, 102]
\end{align*}

Using a brute-force algorithm, we can find another text string that minimizes the number of bit-level changes. In this case, we find that the minimizing text is ``Id: bhavYksikka'', changing the letter in position 8 from ``u'' to ``Y''. The resulting byte array is:
\begin{align*}
b_1 =\text{ }&[64, 244, 150, 67, 162, 6, 38, 134, 23, 101, 150, 183, 54, 150, 182, 182, 16, 236, 17,\\
&234, 95, 209, 74, 163, 175, 89]
\end{align*}

The comparison table between both strings is as follows:

\begin{table}[h!]
\centering
\caption{Binary representation of differing bytes of QR with texts ``Id: bhavuksikka'' and ``Id: bhavYksikka''.}
\setlength{\tabcolsep}{7pt}
\renewcommand{\arraystretch}{1.3}
\begin{center}
\begin{tabular}{c|c|c|c|c|c|}
\cline{2-6}
\multicolumn{1}{c|}{Index}           & 6       & 7       & 8       & 9       & 10       \\ \hline
\multicolumn{1}{|c|}{$b_1$}           & 00100110 & 10000110 & 00010111 & 01100111 &01010110  \\ \hline
\multicolumn{1}{|c|}{$b_2$}           & 00100110 & 10000110 & 00010111 & 01100101 &10010110  \\ \hline
\multicolumn{1}{|c|}{Bits Modified}     & 0        & 0        & 0        & 1        & 2        \\ \hline
\end{tabular}

\vspace*{2 mm}
\begin{tabular}{c|c|c|c|c|c|}
\cline{2-6}
\multicolumn{1}{c|}{Index}          & 16        & 17       & 18       & 19       & 20      \\ \hline
\multicolumn{1}{|c|}{$b_1$}          & 00010000 & 00010000 & 00010000 & 11101011 & 11011111 \\ \hline
\multicolumn{1}{|c|}{$b_2$}          & 00010000 & 11101100 & 00010001 & 11101010 & 01011111 \\ \hline
\multicolumn{1}{|c|}{Bits Modified}    & 0         & 0        & 0        & 1        & 1       \\ \hline
\end{tabular}

\vspace*{2 mm}

\begin{tabular}{c|c|c|c|c|c|}
\cline{2-6}
\multicolumn{1}{c|}{Index}          & 21        & 22       & 23       & 24       & 25       \\ \hline
\multicolumn{1}{|c|}{$b_1$}          & 10010001 & 11011101 & 01001001 & 11101110 & 01100110 \\ \hline
\multicolumn{1}{|c|}{$b_2$}          & 11010001  & 01001010 & 10100011 & 10101111 & 01011001 \\ \hline
\multicolumn{1}{|c|}{Bits Modified}    & 1         & 5        & 5        & 2        & 6        \\ \hline
\end{tabular}
\end{center}
\label{table:comparacion2}
\end{table}

Analyzing Table \ref{table:comparacion2}, we find that we need to make changes in two data bytes corresponding to indices 9 and 10. Additionally, since we have seven error correction codes, we need to fix at least four to evade error correction. Therefore, we select those involving the fewest changes, which are indices 19, 20, 21, and 24. In total, the number of necessary bit-level modifications is eight. The percentage of bits that needs to be changed in the QR code is $\frac{8}{208} * 100 = 3.85\%$.

\begin{figure}[H]
  \centering
  \includegraphics[width=0.8\textwidth]{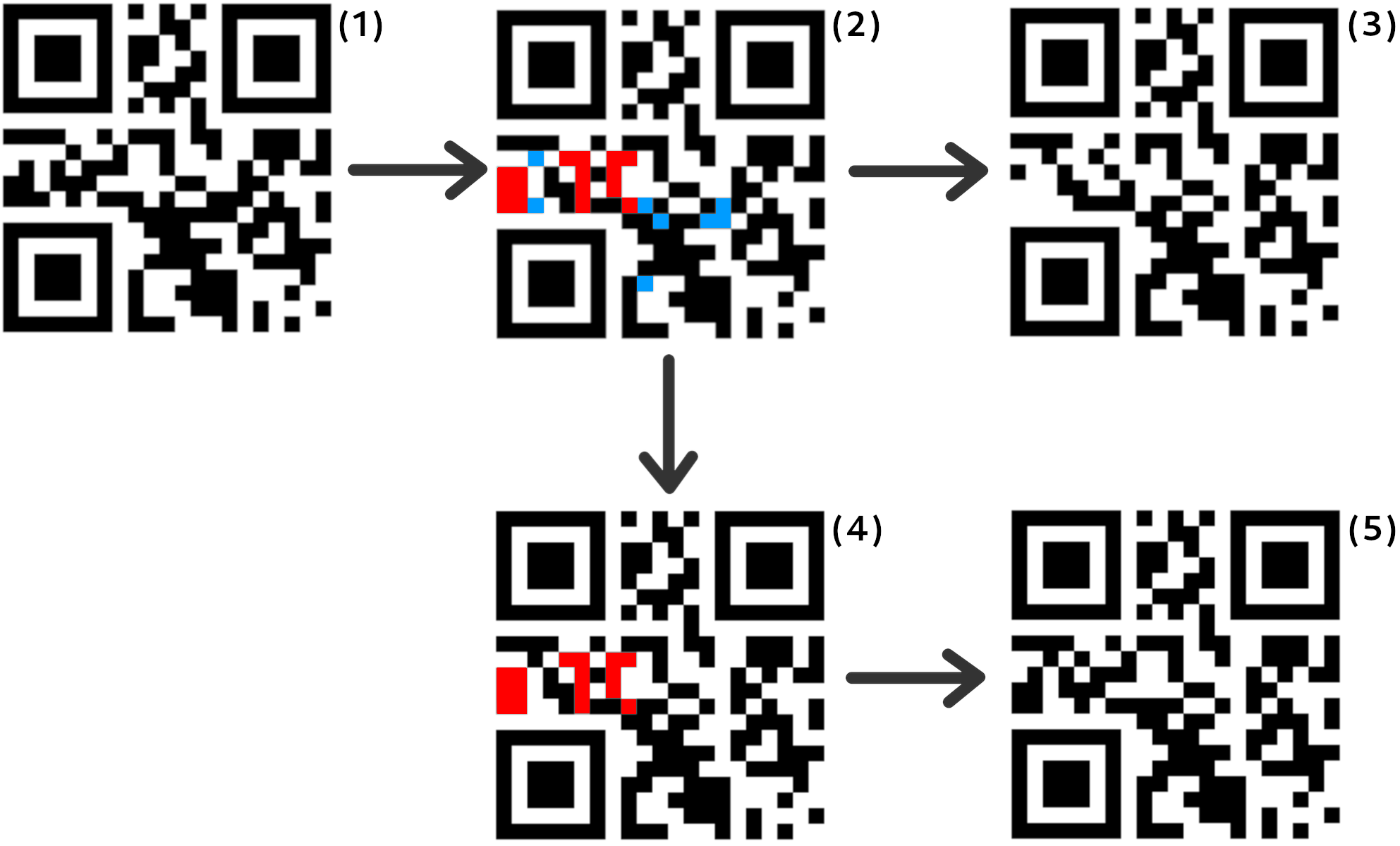}
  \caption{Diagram of the QR code manipulation process in an extreme case.}
  \label{fig:qrproceso2}
\end{figure}

In Figure \ref{fig:qrproceso2}, we observe the same process as in Figure \ref{fig:qrproceso}, modifying the initial and final data to correspond to ``Id: bhavuksikka'' and ``Id: bhavYksikka'' respectively. Here we have marked in red the pixels that don't need changing, and in blue the pixels which need their color flipped.

\subsubsection{Specific Example 2}\label{subseq:example2}

Let's take a look at another example, where the text we want to modify is composed of exactly 17 characters: ``Some binary text.''. Following the same steps as in the previous section, we end up with two possible texts which are the nearest. These are:
\begin{itemize}
    \item Modification in position 14, going from character ``x'' to ``y''. The resulting text is ``Some binary teyt.''.
    \item Modification in position 15, going from character ``t'' to ``4''. The resulting text is ``Some binary tex4.''.
\end{itemize}

Both of these texts require only 7 bit modifications.

\begin{figure}
  \centering
  \includegraphics[width=0.5\textwidth]{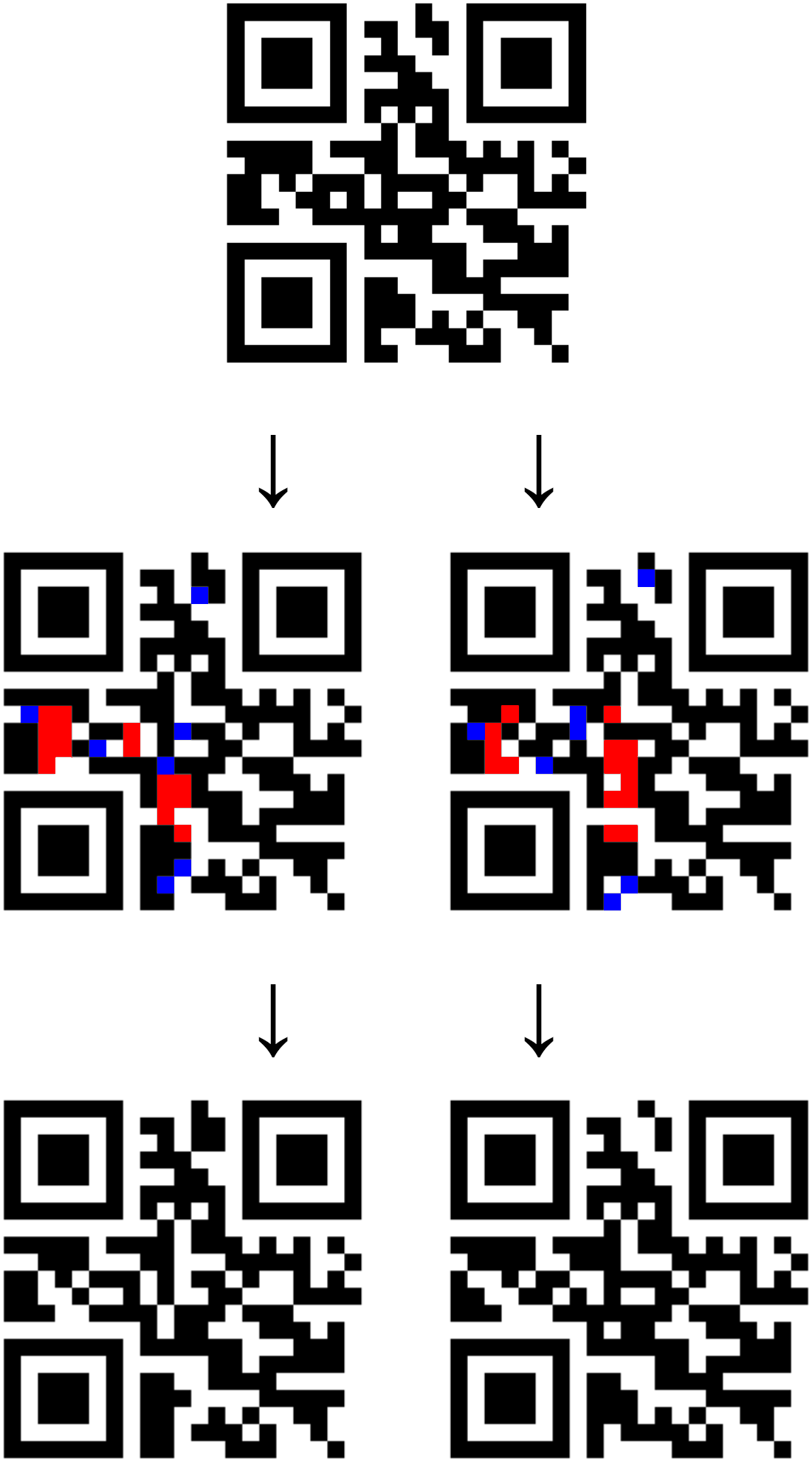}
  \caption{Diagram of the QR code manipulation with two possible outcomes with 7 pixel modifications.}
  \label{fig:qrproceso3}
\end{figure}

In Figure \ref{fig:qrproceso3} we see the flow for the current example. On the left side, we end up with the text ``Some binary teyt.'', and on the right side ``Some binary tex4.''. Following the same color coding as before, blue indicates the pixels that need color changing, and red indicates the pixels that remain the same.

\subsubsection{Generalization}

After numerous tests under the same scenario (using QR version 1, error correction L and data encoded in byte mode), we arrived at the same result regarding the minimum distance to another QR code as in Section \ref{subseq:example2}, which is that there are two possible texts that require 7 modifications to get to.

Additionally, the comparison table of the byte strings between the original QR code and the minimizing QR code is similar to Table \ref{table:comparacion2}, in the sense that the distances at each index are the same; only the values of $b_1$ and $b_2$ change.
\begin{remark}[]
    We have also observed that the closest text is always XORed by a fixed small set of bytes at the same position.
\end{remark}
Following the example in Section \ref{subseq:example2}, the modification at position 14 needs to apply the XOR operation with byte \texttt{0x01}. This means that any QR code using version 1, error correcting L and composed of 15 or more bytes (up to 17) can be modified with only 7 pixel modifications and the modified QR will have the byte at position 14 XORed with byte \texttt{0x01}. 
On the other hand, the modification at position 15 needs to apply the XOR operation with byte \texttt{0x40} or \texttt{0x80}, depending on the original byte at position 15.

As a result of this analysis, we conjecture that fixing the QR version, error correction level and data encoding, the number of bit modifications required to go do from a QR code to the closest QR code is constant.

\begin{table}[H]
\centering
\caption{Number of modifications required to modify the information on a QR code version 1}
\setlength{\tabcolsep}{7pt}
\renewcommand{\arraystretch}{1.3}
\begin{tabular}{|c|c|c|c|c|}
\hline
Error Correcting & Total byte storage & Byte position                                         & Byte XORed                                                        & Bit flips needed \\ \hline
L                & 17                 & \begin{tabular}[c]{@{}c@{}}14\\ 15\end{tabular}       & \begin{tabular}[c]{@{}c@{}}0x01\\ \texttt{0x40}, \texttt{0x80}\end{tabular}               & 7                \\ \hline
M                & 14                 & 1                                                     & \texttt{0x4B}                                                              & 9               \\ \hline
Q                & 11                 & 2                                                     & \texttt{0x0C}                                                              & 14               \\ \hline
H                & 7                  & \begin{tabular}[c]{@{}c@{}}0\\ 3\\ 4\\ 5\end{tabular} & \begin{tabular}[c]{@{}c@{}}\texttt{0x12}, \texttt{0x24}\\ \texttt{0x14}\\ \texttt{0x5B}\\ \texttt{0x61}, \texttt{0xC2}\end{tabular} & 20               \\ \hline
\end{tabular}
\label{table:modif1}
\end{table}

\begin{table}[H]
\centering
\caption{Number of modifications required to modify the information on a QR code version 2}
\setlength{\tabcolsep}{7pt}
\renewcommand{\arraystretch}{1.3}
\begin{tabular}{|c|c|c|c|c|}
\hline
Error Correcting & Total byte storage & Byte position                                  & Byte XORed                                          & Bit flips needed \\ \hline
L                & 32                 & \begin{tabular}[c]{@{}c@{}}3\\ 19\end{tabular} & \begin{tabular}[c]{@{}c@{}}\texttt{0x26}, \texttt{0x4C}\\ \texttt{0x4B}\end{tabular}                                                & 9               \\ \hline
M                & 26                 & 0                                              & \texttt{0x54}, \texttt{0xA8}                                                & 15               \\ \hline
Q                & 20                 & \begin{tabular}[c]{@{}c@{}}2 \\ 6\\ 11\end{tabular} & \begin{tabular}[c]{@{}c@{}}\texttt{0x04} \\\texttt{0x01}\\ \texttt{0x41}\end{tabular} & 23               \\ \hline
H                & 14                 & 4                                              & \texttt{0x88}                                                & 30               \\ \hline
\end{tabular}
\label{table:modif2}
\end{table}

\begin{table}[H]
\centering
\caption{Number of modifications required to modify the information on a QR code version 3}
\setlength{\tabcolsep}{7pt}
\renewcommand{\arraystretch}{1.3}
\begin{tabular}{|c|c|c|c|c|}
\hline
Error Correcting & Total byte storage & Byte position                                  & Byte XORed                                          & Bit flips needed \\ \hline
L                & 53                 & 29                                             & \texttt{0x12}, \texttt{0x24}                                          & 15               \\ \hline
M                & 42                 & \begin{tabular}[c]{@{}c@{}}3\\ 15\end{tabular} & \begin{tabular}[c]{@{}c@{}}\texttt{0xC3}\\ \texttt{0x1F}\end{tabular} & 26               \\ \hline
Q                & 32                 & 3                                              & \texttt{0x9A}                                                & 18               \\ \hline
H                & 24                 & 2                                              & \texttt{0x41}                                                & 23               \\ \hline
\end{tabular}
\label{table:modif3}
\end{table}

In Tables \ref{table:modif1}, \ref{table:modif2} and \ref{table:modif3}, we see the number of modifications required to change a QR code, along with information on what modification would be applied to the resulting text in the QR code.
It's worth mentioning that in Tables \ref{table:modif1} and \ref{table:modif2} there is only one block containing all information for any error correcting level. But in Table \ref{table:modif3} only error correction L and M have a single block; error correction Q and H contain two distinct and independent blocks. In this case, as there are smaller blocks, the total number of modifications is also smaller comparing to the same error correction level in the previous version, in Table \ref{table:modif2}.

\begin{table}[H]
\centering
\caption{Number of modifications required to modify the information on a QR code version 4}
\setlength{\tabcolsep}{7pt}
\renewcommand{\arraystretch}{1.3}
\begin{tabular}{|c|c|c|c|c|}
\hline
Error Correcting & Total byte storage & Byte position                                  & Byte XORed                                          & Bit flips needed \\ \hline
L                & 78                 & 9                                              & \texttt{0x2E}                                       & 19                \\ \hline
M                & 62                 & \begin{tabular}[c]{@{}c@{}}6\\ 9\end{tabular} & \begin{tabular}[c]{@{}c@{}}\texttt{0x15}, \texttt{0x2A}\\ \texttt{0x51}\end{tabular} & 17               \\ \hline
Q                & 46                 & 18                                             & \texttt{0x04}                                       & 28               \\ \hline
H                & 34                 & \begin{tabular}[c]{@{}c@{}}3\\ 4\end{tabular} & \begin{tabular}[c]{@{}c@{}}\texttt{0x10}\\ \texttt{0x01}\end{tabular} & 17               \\ \hline
\end{tabular}
\label{table:modif4}
\end{table}

\begin{table}[H]
\centering
\caption{Number of modifications required to modify the information on a QR code version 5}
\setlength{\tabcolsep}{7pt}
\renewcommand{\arraystretch}{1.3}
\begin{tabular}{|c|c|c|c|c|}
\hline
Error Correcting & Total byte storage & Byte position                                  & Byte XORed                                          & Bit flips needed \\ \hline
L                & 106                 & 4 & \texttt{0x0C}                                            & 25               \\ \hline
M                & 84                 & \begin{tabular}[c]{@{}c@{}}12\\ 13\\ 40\end{tabular} & \begin{tabular}[c]{@{}c@{}}\texttt{0x9C}\\ \texttt{0x0E}\\ \texttt{0x03}\end{tabular} & 24               \\ \hline
Q                & 60                 & 1                                              & \texttt{0x9A}                                                & 18               \\ \hline
H                & 44                 & 0                                              & \texttt{0x41}                                                 & 23               \\ \hline
\end{tabular}
\label{table:modif5}
\end{table}

\begin{table}[H]
\centering
\caption{Number of modifications required to modify the information on a QR code version 6}
\setlength{\tabcolsep}{7pt}
\renewcommand{\arraystretch}{1.3}
\begin{tabular}{|c|c|c|c|c|}
\hline
Error Correcting & Total byte storage & Byte position                                  & Byte XORed                                          & Bit flips needed \\ \hline
L                & 134                   & \begin{tabular}[c]{@{}c@{}}23\\ 29\\ 30\\ 35\\ 42\\ 45\end{tabular} & \begin{tabular}[c]{@{}c@{}}\texttt{0x48}\\ \texttt{0x11}\\ \texttt{0x82}\\ \texttt{0x91}\\ \texttt{0x15}, \texttt{0x2A}\\ \texttt{0x51}\end{tabular} & 17               \\ \hline
M                & 106                   & \begin{tabular}[c]{@{}c@{}}4\\ 12\end{tabular} & \begin{tabular}[c]{@{}c@{}}\texttt{0x48, 0x4C, 0x24}\\ \texttt{0xB8}\end{tabular} & 16               \\ \hline
Q                & 74                   & 16                                             & \texttt{0x03}                                                & 24               \\ \hline
H                & 58                   & 3                                              & \texttt{0x88}                                                & 30               \\ \hline
\end{tabular}
\label{table:modif6}
\end{table}

On tables \ref{table:modif4}, \ref{table:modif5} and \ref{table:modif6} we have the same analysis applied to QR codes versions 4, 5 and 6 respectively.

\begin{table}[H]
\centering
\caption{Number of modifications required to modify the information on a QR code version 40}
\setlength{\tabcolsep}{7pt}
\renewcommand{\arraystretch}{1.3}
\begin{tabular}{|c|c|c|c|c|}
\hline
Error Correcting & Total byte storage & Byte position                                   & Byte XORed                                          & Bit flips needed \\ \hline
L                & 2953               & 51                                              & \texttt{0x91}                                                & 28               \\ \hline
M                & 2331               & 34                                              & \texttt{0x88}                                                & 30               \\ \hline
Q                & 1663               & \begin{tabular}[c]{@{}c@{}}16\\ 20\end{tabular} & \begin{tabular}[c]{@{}c@{}}\texttt{0x71}\\ \texttt{0xF8}\end{tabular} & 32               \\ \hline
H                & 1273               & \begin{tabular}[c]{@{}c@{}}7\\ 11\end{tabular}  & \begin{tabular}[c]{@{}c@{}}\texttt{0x71}\\ \texttt{0xF8}\end{tabular} & 32               \\ \hline
\end{tabular}
\label{table:modif40}
\end{table}

As we know, with bigger QR codes, the information to be shown is separated in different blocks in order to limit the complexity of encoding and decoding. In fact, the limit imposed is that a block can only have up to 30 error correcting codes. This limitation results in an upper bound on the number of bit flips required to change the information on a QR code. In Table \ref{table:modif40} we see the number of bit flips needed to change the information using the highest QR code version, setting the upper bound to 32 bit flips.

In version 40, block sizes vary depending on the error correcting level as follows:
\begin{itemize}
    \item Low: 
    \begin{itemize}
        \item 19 blocks with 118 bytes of data and 30 bytes of error codes
        \item 6 blocks with 119 bytes of data and 30 bytes of error codes.
    \end{itemize}
        \item Medium: 
    \begin{itemize}
        \item 18 blocks with 47 bytes of data and 28 bytes of error codes
        \item 31 blocks with 47 bytes of data and 28 bytes of error codes.
    \end{itemize}
        \item Quartile: 
    \begin{itemize}
        \item 34 blocks with 24 bytes of data and 30 bytes of error codes
        \item 34 blocks with 25 bytes of data and 30 bytes of error codes.
    \end{itemize}
        \item High: 
    \begin{itemize}
        \item 20 blocks with 15 bytes of data and 30 bytes of error codes
        \item 61 blocks with 16 bytes of data and 30 bytes of error codes.
    \end{itemize}
\end{itemize}

\section{Conclusions}
In this work, we have conducted a thorough analysis of the underlying mathematical construction in the generation of QR codes, addressing multiple key aspects that influence their operation and security.

Firstly, we have explored the fundamental theoretical principles that underpin the creation of QR codes, focusing on their modular structure and the algebraic properties that define them.

Results have shown a vulnerability of QR codes to selective data editing, which can be problematic in contexts where information integrity is essential, such as in QR code-based payment methods. This study indicates that any QR code using byte mode encoding can have one byte modified with no more than 32 pixel changes. In the worst case scenario, only 7 pixel changes result in a modification in a QR code in version 1 and error correction L.

The code used for the analysis can be found at \url{https://github.com/Bubbasm/tfgmates}.

\subsection{Improvements proposed}

A proposed improvement is to use error correction level H, requiring many modifications in order to change the data. This approach has the caveat of needing very big QR codes to represent any information, and even then, the amount of pixel modifications tops out to 32 in this case.

Another approach would be to implement a system where the QR reader has the capability to verify if the mask used in the QR code is the most optimal for ensuring readability and data integrity. Nonetheless, this would imply a higher computational cost for the decoder, which is not always feasible. Moreover, there may be cases of QR modification that do not use a different mask, so this would not be a good solution.

Additionally, it is important to note that these proposed improvements cannot solve the underlying fundamental problem. The true challenge lies in the intrinsic nature of how a QR code is created, where words are formed using 8 bits. This allows each word to be broken down into individual bits, which simplifies the process of selective modification and the introduction of minimal changes that can go unnoticed by detection systems.

One possible solution would be to use bit-level error correction for the entire message. This is already done in the format bits found in a QR code, as we saw in section \ref{qrstages}.
However, the encoding and decoding of these codes would be considerably more costly, as we would be dealing with much larger polynomial divisions than the simplified versions introduced in example \ref{ejemp:algoritmocodif}.

Consequently, it is essential to consider more comprehensive and sophisticated approaches to address these security issues in QR codes, recognizing the inherent limitations of their fundamental design and structure.

\bibliographystyle{unsrtnat}
\bibliography{references}

\end{document}